\documentclass{article}
\usepackage{graphicx} 

\usepackage[utf8]{inputenc}
\usepackage[english]{babel}
\usepackage{mathtools}
\usepackage{nccmath}
\usepackage{amsmath,amssymb}
\usepackage{graphicx} 
\usepackage[utf8]{inputenc}
\usepackage[english]{babel}
\usepackage{mathtools}
\usepackage{nccmath}
\usepackage{amsmath,amssymb}
\usepackage{graphicx}
\usepackage{subcaption}
\usepackage[margin=1in]{geometry}
\usepackage{natbib}
\usepackage{float}
\usepackage{comment}
\usepackage{soul}
\usepackage{authblk}
\usepackage{url}
\usepackage{hyperref}

\usepackage{appendix}
\usepackage{amsthm}
\usepackage{xcolor}
\newtheorem{theorem}{Theorem}[section]

\title{Pulse desynchronization of neural populations by targeting the centroid of the limit cycle in phase space}

\author[1,2,3]{Ram\'on Guevara}
\author[3]{Marco Zenari}
\author[4]{Giorgio Nicoletti}
\author[5]{Elisa Marini}
\author[1,3,6]{Samir Suweis}
\author[1,6]{Sandro Azaele}
\author[3,7]{Marco Formentin}
\affil[1]{Department of Physics and Astronomy, University of Padua, Padua, Italy}
\affil[2]{Department of Developmental and Social Psychology, University of Padua, Padua, Italy}
\affil[3]{Padua Neuroscience Center, Padua, Italy}
\affil[4]{Abdus Salam International Center for Theoretical Physics, Trieste, Italy}
\affil[5]{CEREMADE, UMR CNRS 7534, Université Paris Dauphine-PSL, Paris, France}
\affil[6] {INFN, Padua section, Padua, Italy}
\affil[7]{Department of Mathematics, University of Padua, Padua, Italy}

\begin{document}
\maketitle

\begin{abstract}
The synchronized activity of neuronal populations can lead to pathological over-synchronization in conditions such as epilepsy and Parkinson disease. Such states can be desynchronized by brief electrical pulses. But when the underlying oscillating system is not known, as in most practical applications, to determine the specific times and intensities of pulses used for desynchronizaton is a difficult inverse problem. Here we propose a desynchronization scheme for neuronal models of bi-variate neural activity, with possible applications in the medical setting. Our main argument is the existence of a peculiar point in the phase space of the system, the centroid, that is both easy to calculate and robust under changes in the coupling constant. This important target point can be used in a control procedure because it lies in the region of minimal return times of the system.        
\end{abstract}

\section{Introduction}

The adjustment of rhythms of self-sustained or weakly chaotic oscillators due to weak coupling is a widespread phenomenon in nature, with examples ranging from physics to chemistry, biology (e.g., pacemaker cardiac cells, Josephson junctions, lasers, mechanical vibrations, electronic circuits, metronomes, fireflies) (\cite{pikovsky2002synchronization}). Understanding network synchronization is particularly relevant for neuroscience (\cite{ermentrout2010mathematical}) because the collective coordinated (synchronized) activity of neuronal populations is at the origin of neural rhythms (\cite{buzsaki2006rhythms}). Furthermore, neural synchronization is believed to play a fundamental role in brain functionality, as a mechanism for information integration in perceptual and cognitive processes (\cite{buzsaki2006rhythms, schnitzler2005normal, varela2001brainweb}). 

The malfunction of synchronization mechanisms may also affect brain functionality. It is well known that excessive neural synchronization is often associated with pathological activity, particularly in epilepsy, Parkinson's disease, and essential tremor (\cite{jiruska2013synchronization, park2011neural, batista2010delayed, milton2003epilepsy}). For example, during an epileptic seizure, the activity of neuronal networks is highly synchronized across brain areas, leading to a well-recognized high amplitude and fairly rhythmic electroencephalographic (EEG) signal. For this reason, the last few decades have witnessed increased interest in the monitoring and control of neuronal synchronization, with the ultimate goal of suppressing pathological over-synchronized brain states.
Recent technological advances have provided a way to alleviate such pathological conditions by low-intensity electric stimulation of the affected brain regions or peripheral nerves. This is particularly useful for patients who do not respond to drug medications. In this sense, one of the most important breakthroughs was the development of deep brain stimulation (DBS), a technique to suppress synchronized activity in epilepsy and Parkinson disease (\cite{benabid1991long, kringelbach2007translational,kuhn2017innovations}). In DBS, micro-electrodes implanted in the brain locally stimulate neuronal tissue by application of high frequency electric pulses (around 120 Hz), effectively modulating neuronal dynamics. The mechanism of action of DBS is not yet fully understood (\cite{johnson2008mechanisms, deniau2010deep}). However, it has been hypothesized that it suppresses pathological oscillatory activity through the desynchronization of neuronal populations, without suppressing the oscillatory activity of individual neurons (\cite{tass2007phase}).
Due to the negative effects that injected currents can induce in patients (such as speech impairment and involuntary muscle contraction) (\cite{okun2010parkinson,skodda2012effect}), the current goal of DBS and similar techniques is to achieve a minimally invasive stimulating signal, that is, to achieve desynchronization of neuronal populations with a minimal number of injected pulses. 
 
Two fundamental approaches to achieve desynchronization of collective neural oscillations have been proposed in the literature: 1) those based on phase resetting techniques (\cite{tass2001effective, tass2003model,tass2007phase, hauptmann2009cumulative,lysyansky2011desynchronizing}) and the injection of small pulses of current; 2) Those originating from the control theory community, using continuous close-loop feedback stimulation (\cite{pikovsky2004controlling, rosenblum2004delayed, montaseri2013synchrony}). The latter approach, based on the idea of using a closed loop feedback control of synchronization, has acquired great credibility due to its effectiveness in theoretical models and simulations. In fact, it offers the possibility of implementing a vanishingly small stimulation signal, a crucial feature in medical applications (as we mentioned above, DBS can lead to collateral damages, so it is important to minimize the stimulating signal) (\cite{tukhlina2007feedback,montaseri2013synchrony}).
In the medical setting, the closed loop feedback stimulation concept is also gaining credibility and has been implemented with relative success (\cite{rosin2011closed,little2013adaptive}).
In more recent years, the idea of a desynchronization of neuronal populations by closed-loop feedback using stimulating pulses instead of continuous signals (that is, a combination of both approaches) was developed and applied to neuronal models, showing that it is possible to desynchronize the output activity of the models by very small-amplitude pulses (\cite{rosenblum2020controlling,krylov2020reinforcement}). 
In the current study, we follow a similar approach, that of synchronization control through pulse stimulation, but with two main differences from previous work (\cite{rosenblum2020controlling,krylov2020reinforcement}). 

The first difference has to do with the input signal (input to the feedback-control system). That is, what kind of neurophysiological signal does our device measure? In previous literature, the input is univariate, but we consider the case where a bi-variate signal is measured. Indeed, it is typically assumed that the input signal is the local field potential (average voltage of the neuronal population, a univariate time series) or a related signal, as is customary in electrophysiology. Instead, here we assume that a bi-variate time series is measured, which is justified by  recent advances in neurophsyological and multimodal measurements, allowing for bi-variate, and more generally multivariate recordings. For example, at the level of local circuits, it is possible to simultaneously record neuronal spikes and calcium imaging in vivo (\cite{wu2021modified}), an example that would constitute a bi-variate input signal for a control device. Recent works have also shown that it is possible to measure excitation and inhibition (that is, a bi-variate signal) simultaneously, on the basis of electrophysiological recordings (\cite{muller2021novel,bruining2020measurement}).
Another example of a multivariate signal is the simultaneous measurement of brain activity using EEG and MEG (although it is worth remembering that EEG and magnetoencephalography (MEG) signals are only partially independent) (\cite{malmivuo2012comparison}).  
We highlight that, at the whole brain scale, especially when such non-invasive techniques (scalp EEG or MEG) are used, oscillatory activity is described in terms of phenomenological models that do not make explicit use of individual neurons (\cite{glomb2021computational}). In this direction, Wilson-Cowan population models have been extensively used to model medium and large-scale brain activity in terms of excitation and inhibition (\cite{destexhe2009wilson, byrne2020next}). Such models have been successfully used to describe epileptic seizures (\cite{meijer2015modeling, lytton2008computer}). 

Second, the control problem is addressed in previous work through successive approximations; that is, by an adaptive search for the most effective phase for pulse stimulation (\cite{rosenblum2020controlling}). This strategy requires applying many pulses before the population signal is sufficiently small. Here, we exploit the fact that the input signal is bi-variate to extract information about the phase space and other dynamic and geometric properties of the underlying dynamical system, without completely fitting it to the data. Knowledge of such properties gives us the opportunity to achieve synchronization control with a much lower number of pulses, a feature that could be useful in the medical setting, minimizing the invasiveness of the method and collateral damage. 

\section{Theoretical Framework}

\subsection{Modeling pathological neural activity}

The neural activity observed in certain types of pathologies (such as epileptic seizures and Parkinson's tremor) is strongly periodic and synchronized. As explained above, we assume that the measuring device records a bi-variate $(x,y)$ signal. To maintain some level of generality, we do not specify here the nature of this signal (as already mentioned, it could be a combination of two univariate signals obtained with different techniques). 
To model such pathological neural activity, we consider a FitzHugh-Nagumo (also known as Bonhoeffer - van der Pol) system of $N$ oscillators (units) coupled through the mean field $X =  \frac1N\sum_{j=1}^N x_j$:  
\begin{equation}\label{eq:FitzHugh-Nagumo}
    \begin{split}
        \frac{dx_i}{dt} & = \xi x_i + \delta x_i^3 + \nu y_i + \epsilon X   + I_i\\
        \frac{dy_i}{dt} & = \alpha \left( x_i + \beta y_i + \gamma \right)
    \end{split}
\end{equation}
where $x_i (t)$ and $y_i (t)$ are the time-dependent activities of each interacting oscillating unit and $\delta,\xi\,\,\nu,\, \epsilon,\, \alpha,\,\beta$ and $\gamma$ are parameters. For simplicity and following the literature (\cite{pikovsky2004controlling,rosenblum2020controlling,krylov2020reinforcement}), we also assume that the interaction of an individual oscillator with the mean field generated by the other  units occurs only in the $x$-dimension (see \cite{pikovsky2004controlling, rosenblum2020controlling}), through the term $\epsilon X$, where $\epsilon$ is a (weak)  coupling parameter. 

The $x$-variable is also coupled with a current, $I_i$, different for each oscillator and sampled from a Gaussian distribution with mean $\mu$ and variance $\sigma^2$.
We highlight three main points about the proposed modelling framework.
1) The oscillators do not necessarily represent individual neurons; rather, they can model the activity of a large ensemble of functionally similar neurons (e.g. resting state networks). The collective neural oscillation emerges as a result of the interaction between these coupled oscillators. 2) The interaction between oscillators is mathematically described in a very compact form: each individual oscillator is coupled to others through the ensemble average (mean field, $X$). The coupling of each oscillator to the average is assumed to be weak (this assumption is typically satisfied for neuronal networks (\cite{hoppensteadt1997weakly})). 3) The coupling (interaction) term depends only on one variable. 

Rather than aiming for biological realism or an exact match to known neuronal oscillations, we focus on the fact that each oscillator has an associated limit-cycle, a property that we exploit in the control strategy, as explained below. In Eq.~\eqref{eq:FitzHugh-Nagumo}, the emergence of synchronized oscillations is possible within a certain range of parameters (\cite{kuramoto_chemical_2003, pikovsky2002synchronization}). For concreteness and visualization, let us consider the following choice of parameters: $\xi=-\nu=1$, $\delta = -1/3$, $\alpha=0.1$, $\beta=-0.8$, $\gamma = 0.7$, $\mu=0.6$ and $\sigma^2 = 0.01$.  In this case, as shown in Fig.~\ref{fig:FitzNagumo}, each oscillator displays a limit cycle and oscillates periodically at a certain frequency. The average (mean-field) activity, $X(t)$, is also a periodic signal with a period similar to the period of each individual oscillator. The currents $I_i$ are all different; therefore, there are slight differences in the shape of the limit cycle and the frequency of each oscillator, with the parameter $\sigma$ (standard deviation of $I$) characterizing the shapes of the cycles and the spread in the frequency of oscillation of the population.
With our choice of parameters, Eq.~\eqref{eq:FitzHugh-Nagumo} represents a good phenomenological model of the emergence of the macroscopic neural oscillations observed during pathological conditions. 

The parameter $\epsilon$ serves as a global coupling; increasing it leads to the synchronization of the population. Therefore, it can model the onset of the pathology (for large enough $\epsilon$). In Fig.~\ref{fig:FitzNagumo}, we show simulations of the system with the choice of parameters given above, and with $\epsilon$ running from 0 to 0.5. At values of $\epsilon$ very close to 0, the population is desynchronized, with each individual unit having a different phase and a slightly different frequency. Increasing $\epsilon$ leads to a decrease in the population spread in phase space due to synchronization. 

\subsection{Feedback signal and pulse-control strategy}

In a pulse feedback control paradigm, the neural system is located inside a closed-loop: a physiological signal recorded from the neural system is constantly measured and used as input signal for the controlling apparatus that operates on it, in real time. The output signal of the controlling apparatus is in turn used to modulate the timing and intensity of the pulses injected back to the neural system.   
The main difference with the previous literature, at this point, is that we assume knowledge of both $X = \frac{1}{N}\sum_{j=1}^N x_j$ and $ Y = \frac{1}{N}\sum_{j=1}^N y_j$, that is, we assume that the recording of neural activity is a bi-variate time series (whereas in previous publications it was assumed knowledge of $X$ only, see e.g. (\cite{pikovsky2004controlling,rosenblum2020controlling})). In other words, the inputs for the feedback loop are $X(t)$ and $Y(t)$. This signal is assumed to be monitored at all times. To mimic experimental conditions, we also assume that we have no knowledge of the dynamics of individual units. In summary, we know $X$ and $Y$ at all times, but we do not know the individual  $x_i$ and $y_i$ values, nor the exact form of the Eq.~\eqref{eq:FitzHugh-Nagumo}.

As for the feedback pulse desynchronization control strategy, we suggest a model in which pulses are given to all the population elements simultaneously and with equal intensity. We propose the following simplification of the effect of a pulse of current on the system: after the pulse, the population averages $X$ and $Y$ are changed; that is the effect of the pulse. In other words, a pulse is equivalent to resetting the system, after which the system is restarted at new initial conditions (this simulates the \lq\lq kick\rq\rq\, on the system due to the current pulse, a discontinuous process). Conceiving the dynamics of the neural system as movement in the ($x$,$y$) plane, the pulse effect is that of arbitrarily changing the system averages ($X$,$Y$) to a new location in that plane. We do not provide pulses with different intensities to each individual unit, but to the whole population average. This is because we assume that we cannot monitor or control individual units, so we model the effect of the pulse on the average, making the approximation that the effect of the pulse is the same for all units.  We also assume that we have the ability to give pulses that change both the $x$ and the $y$ coordinates. So, if the pulse of intensity ($A$,$B$) is given at time $t_p$, then we stop the simulation at this time and establish new initial conditions by pushing the system instantaneously to new coordinates: $x_i(t_p) \rightarrow x_i(t_p) + A$ and $y_i(t_p) \rightarrow  y_i(t_p) + B$.

What is the strategy followed in our pulse feedback-loop paradigm? When are pulses given and with what strength? Our logic is that it is convenient to push the average's system (X and Y variables) into specific points or regions of its phase space. To provide an example from the existing literature, in (\cite{rosenblum2020controlling}), the pulses are presented in such a way as to place the average system in the vicinity of an instability region (which is identified by making $\epsilon = 0$ in Eq.~\eqref{eq:FitzHugh-Nagumo}). Hence, in comparison with that work, our task is facilitated by the fact that we assume knowledge of both $X$ and $Y$ (whereas in (\cite{rosenblum2020controlling}) knowledge of $X$ alone is assumed).

\begin{figure}
    \centering
    \includegraphics[width=0.9\linewidth]{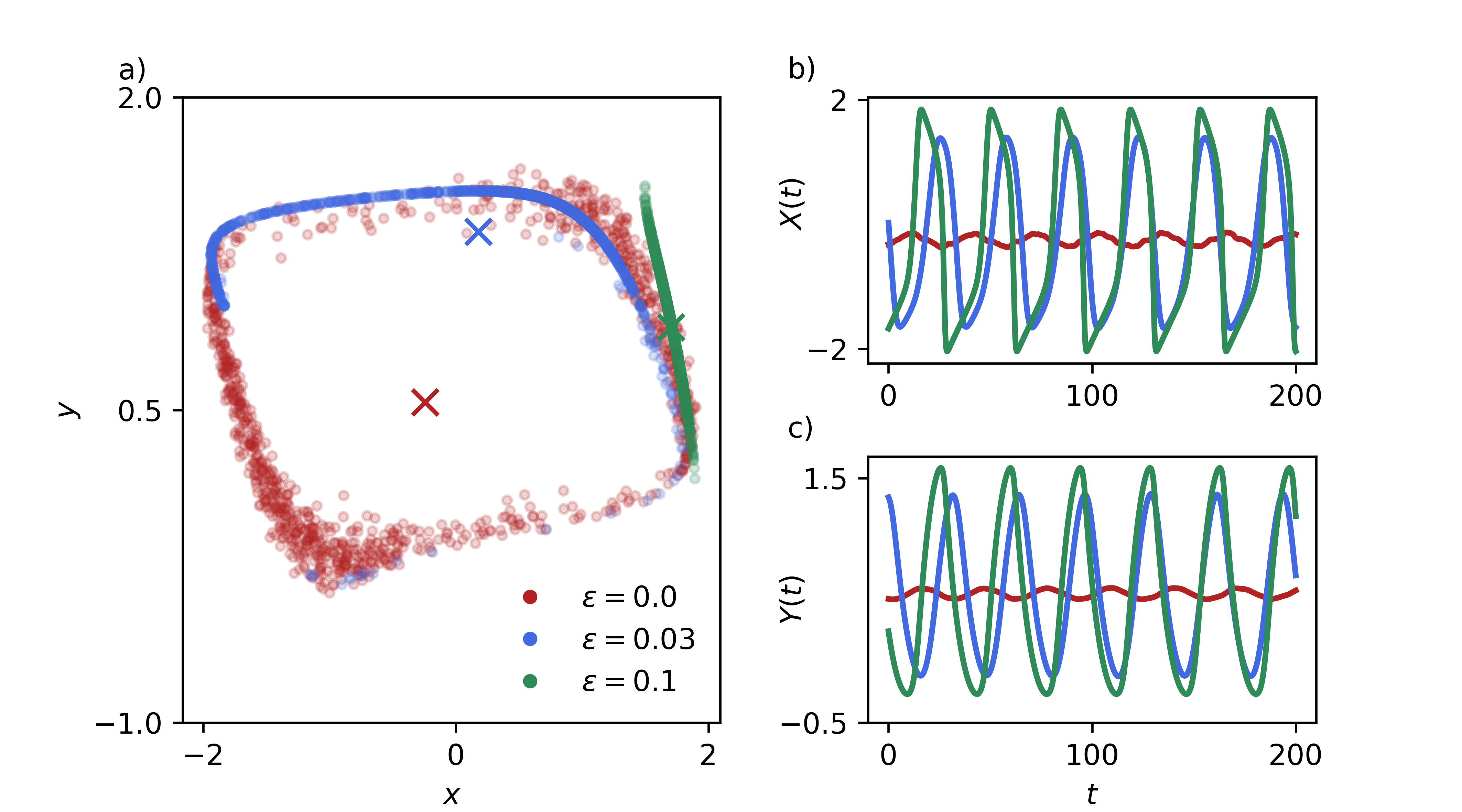}
    \caption{\textbf{FitzHugh--Nagumo (Bonhoeffer--van der Pol) system.} a) Example configurations of $N = 1000$ oscillators in the $(x,y)$ phase space at a fixed time, shown for three different values of the coupling $\varepsilon$. All other parameters are fixed to
$\xi = 1,\ \nu = -1,\ \delta = -\tfrac{1}{3},\ \mu = 0.6,\ \alpha = 0.1,\ \beta = 0.8,\ \gamma = 0.7,\ \sigma = 0.1.$
Crosses indicate the position of the ensemble mean $(X, Y)$ averaged in time. 
b) Temporal evolution of the ensemble mean $X(t)$.
c) Temporal evolution of the ensemble mean $Y(t)$.
 }
    \label{fig:FitzNagumo}
\end{figure}

\section{Results}

\subsection{Pulse control on the basis of the centroid of the limit cycle in phase space}

Numerical simulations of Eq.~\eqref{eq:FitzHugh-Nagumo} indicate that, in the phase space $(x,y)$, a single reference point can be used to desynchronize the population. We take this point to be the centroid (geometric center) of the oscillators' limit cycles, which are nearly coincident in phase space, and define it as follows.

First, consider the limit cycle of an individual oscillator simply as a geometric object in the 2D phase space (a closed curve; see Fig.~\ref{fig:FitzNagumo}). The \,\emph{centroid} of this curve is the point $(x_c,y_c)$ we use for control. Geometrically, it is the same point you would obtain as the balance point of a closed wire of uniform density shaped like the limit cycle (equivalently, of the region enclosed by the cycle filled with uniform density, depending on the chosen convention). This should not be confused with the time average of the trajectory, i.e., formally,  $\lim_{t\to +\infty}\left( \frac{1}{t}\int_0^t x_i(s) ds, \frac{1}{t}\int_0^t y_i(s) ds\right)$: because the oscillator moves with non-uniform speed along the cycle, the trajectory spends more time in some segments than in others, so the time-averaged point generally differs from the geometric centroid. 
In Fig.~\ref{fig:FitzNagumo}, another curve important for our argument is also shown: the set of points $(X(t),Y(t))$, that is, the mean-field equivalent to the limit cycles. We call these curves \lq\lq mean-field cycles\rq\rq.
We highlight that the centroids of limit cycles corresponding to different values of $\epsilon$ are almost coincident (they are nearly indistinguishable at the scale used in Fig.~\ref{fig:FitzNagumo}). Therefore, the centroid is almost independent of the value of $\epsilon$. Furthermore, the centroids lie within the set of possible mean-field cycles; that is, even for very small values of $\epsilon$, the mean-field cycles evolve around the centroids. 

These two properties can be used for pulse control. In applications, the value of $\epsilon$ is not known, so it is useful that the centroid is almost independent of $\epsilon$. The importance of this fact is that once the centroid of the data (at a given $\epsilon$) is found, the centroids at any other values of $\epsilon$ are very close to the one found. It is also useful that the centroid lies within the mean-field cycles for all values of $\epsilon$. The implication is that, pushing the system in the centroid, it is automatically very close to the mean-field cycle at $\epsilon = 0$. This can be seen in Fig.~\ref{fig:FitzNagumo}: if the system is pushed to the centroid position, then it is very close to the centroid at $\epsilon = 0$ and inside the mean-field cycle curve for the smallest value of $\epsilon$ used.

\subsection{Return times to the limit cycle}

To assess the effectiveness of the pulse feedback control procedure, we investigated how long it takes for the system to return to the mean-field limit cycle after being perturbed by a pulse. The reasoning is as follows. When the pulse control is applied to a population already on the limit cycle, each oscillator is displaced in phase space to a new point. This displacement is rigid, meaning that all individual oscillators are shifted equally in phase space. The main question is, therefore, how long it takes for a perturbed population to return to the limit cycle. In particular, our goal is to identify points in phase space where the return time is maximal. Longer return times indicate a more effective control procedure.
We proceed in two steps. First, we numerically explore the return times of the system to the limit cycle. 
Subsequently, we exploit the thermodynamic limit of system \eqref{eq:FitzHugh-Nagumo} and a few approximations derived from the simulations to analytically characterize the effectiveness of the proposed control strategy.

\subsubsection{Numerical investigation of pulse control strategy}

To numerically explore the effectiveness of the proposed pulse feedback control strategy, and in particular the choice of the control point, we performed extensive simulations to estimate the return time of the system to the mean-field limit cycle after a control pulse perturbation. These simulations are intended as an exploratory investigation, providing insights into the behavior of the system and guiding further analytical considerations.

In the simulations, we consider a system of $N=1000$ oscillators of the form given in Eq.~\eqref{eq:FitzHugh-Nagumo} with parameters $\xi = 1$, $\nu = -1$, $\delta = -1/3$, $\mu = 0.6$, $\alpha = 0.1$, $\beta = 0.8$, $\gamma = 0.7$, $\sigma = 0.1$, and different values of $\epsilon \in \{0.1, 0.2, 0.3, 0.4\}$.  

The pulse perturbation is implemented as a translation in phase space, shifting the population average $(X,Y)$ from a point on the limit cycle, where the system is synchronized, to another point within the limit cycle. This new point serves as the initial condition for the subsequent simulation, during which the system relaxes back to the limit cycle. The spatial distribution $(x_i, y_i)$ around the mean $(X, Y)$ is also slightly perturbed by Gaussian noise with standard deviation $\sigma = 0.1$.

After the perturbation, the system evolves freely according to the dynamical Eq.~\eqref{eq:FitzHugh-Nagumo}, and the first return time to the limit cycle is recorded. A trajectory $(X,Y)$ is considered to have returned to the mean-field limit cycle if its distance to the cycle is less than a threshold $dr$, chosen empirically as $dr = 0.03$.  
 
\begin{figure}[ht]
    \centering

    \begin{subfigure}[b]{0.45\textwidth}
        \centering
        \includegraphics[width=\textwidth]{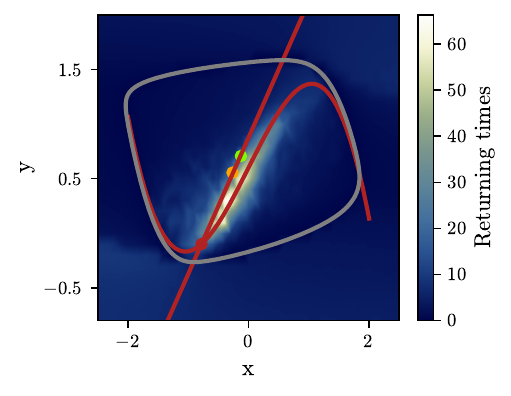}
        \caption{$\varepsilon = 0.1$}
    \end{subfigure}
    \hfill
    \begin{subfigure}[b]{0.45\textwidth}
        \centering
        \includegraphics[width=\textwidth]{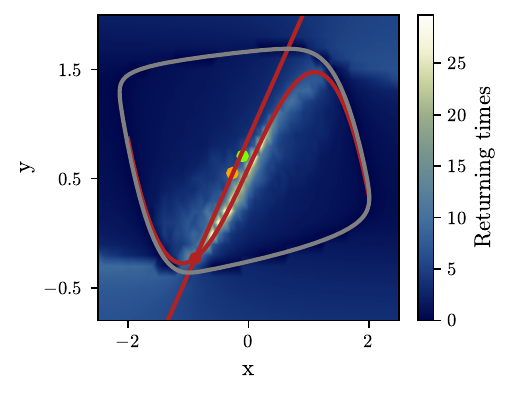}
        \caption{$\varepsilon = 0.2$}
    \end{subfigure}

    \par\medskip

    \begin{subfigure}[b]{0.45\textwidth}
        \centering
        \includegraphics[width=\textwidth]{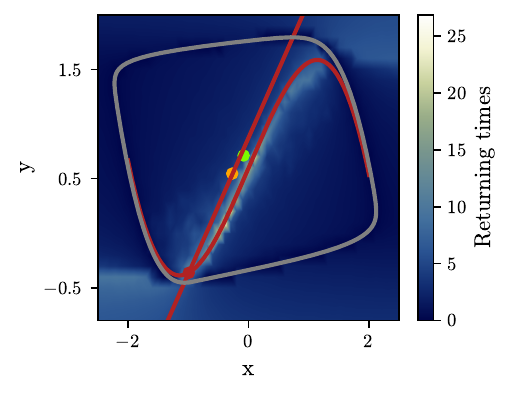}
        \caption{$\varepsilon = 0.3$}
    \end{subfigure}
    \hfill
    \begin{subfigure}[b]{0.45\textwidth}
        \centering
        \includegraphics[width=\textwidth]{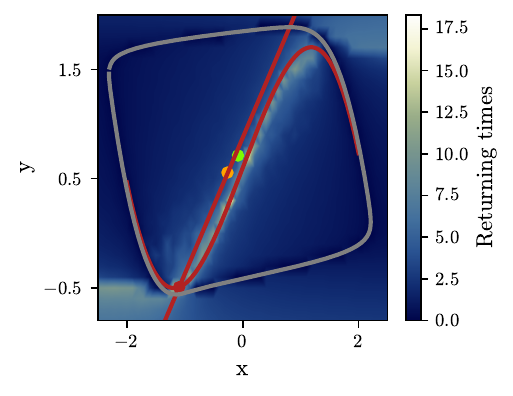}
        \caption{$\varepsilon = 0.4$}
    \end{subfigure}

    \caption{
    \textbf{Interpolated heatmap of return times.} 
    The mean-field limit cycle is shown as a solid gray line; nullclines as solid red lines; the unstable fixed point as a red marker.
    The geometric center of the mean-field limit cycle is shown in green, while the averaged-in-time $(X,Y)$ position for $\epsilon = 0$ is shown in orange. All geometric curves and points are derived from the averages of $X$ and $Y$. Simulations were performed with $N=1000$ oscillators and $1479$ perturbation points on a grid with spacing $\mathrm{d}x=\mathrm{d}y=0.1$.
    Return times were linearly interpolated within the grid.
    For each perturbation point, $100$ trials were run for a simulation time of $100$ with timestep $dt=0.05$.}
    
    \label{fig:return_times}
\end{figure}

To investigate which points in phase space are most relevant for control, we constructed a square grid inside the limit-cycle region with spacing $\mathrm{dx} = 0.1$, $\mathrm{dy} = 0.1$. Each grid node was used as a starting point for a simulation, mimicking a pulse that displaces the system inside the limit cycle and then allows it to evolve back. For each point, the corresponding return time was recorded. The results are shown in Fig.~\ref{fig:return_times}, where the return time is encoded in color at each grid point. The figure also shows relevant geometric features, including nullclines, the unstable fixed point, the geometric center of the limit cycle, and the averaged-in-time $(X, Y)$ position for $\epsilon = 0$ at stationarity. 

Several insights can be gained from these exploratory simulations. 

First, the return times are largest (yellow regions) when the perturbation occurs near the $\frac{dx}{dt}=0$ nullcline (red hyperbolic curve). This can be attributed to the system having two characteristic times: fast motion along the upper and lower branches of the mean-field limit cycle, and slow motion along the right and left branches. Starting at a point with zero initial $x$ velocity delays the return to the limit cycle, which primarily occurs along the $x$ direction. The region of maximal return times intersects the $\frac{dy}{dt}=0$ nullcline only near the center of the figure.  

Second, both the geometric center and the averaged-in-time position of $(X,Y)$ at $\epsilon = 0$ are close to the region of longest return times, with the $x$-nullcline lying nearby. The unstable fixed point, located at the intersection of both nullclines, corresponds to a region of minimal return times.  

Third, as $\epsilon$ increases, the region of short return times narrows around the $x$-nullcline.

In conclusion, numerical simulations show that the centroid is a feasible point for controlling synchronization because it is inside the region of minimal return times.

\subsubsection{Thermodynamic limit and analytical estimation of return times}

In the previous section, the return times after pulse stimulation were investigated numerically. It was found that targeting the centroid for synchronization control is a good strategy because the centroid is in a region of large return times. 

In the current section, we describe an analytical argument in support of this idea.

The main idea is to investigate the dynamics of interacting FitzHugh-Nagumo oscillators (described by Eq.~\eqref{eq:FitzHugh-Nagumo}) in the thermodynamic limit, when the number of oscillators becomes very large ($N\to +\infty$).  
Indeed, system \eqref{eq:FitzHugh-Nagumo} has the propagation of chaos property (\cite{touboul12, touboul2014propagation, reviewPoCI, nodea}): in the limit $N\to +\infty$, the oscillators  effectively behave as independent units, all obeying the  equation

\begin{equation}\label{eq:FitzHugh-Nagumo_Mean_Field}
    \begin{split}
        \frac{d \hat x}{dt} & = \xi \hat x + \delta \hat{x}^3 + \nu \hat y + \epsilon \mathbb{E}[\hat x] + \mu + \sigma I,\\
        \frac{d \hat y}{dt} & = \alpha \left( \hat x + \beta \hat y + \gamma \right),
    \end{split}
\end{equation}
where $I$ is a standard Gaussian random variable with zero mean and unit variance ($I \sim \mathcal{N}(0,1)$). This term represents the effect of the random currents $I_i$ in the original oscillator population (quenched noise). The expectation $\mathbb{E}[\cdot]$ is taken with respect to the joint law of $(\hat x, \hat y)$ at a given time $t$. The rigorous statement and proof of the propagation of chaos property are postponed to Appendix \ref{app}.

In this thermodynamic limit, the mean-field equations \eqref{eq:FitzHugh-Nagumo_Mean_Field} allow us to analyze the macroscopic dynamics of the system analytically and to explore the influence of control pulses on the collective behavior of the oscillators without simulating the full high-dimensional population. 
Indeed, the propagation of chaos amounts to saying that a law of large numbers holds for $(X(t), Y(t))$, which means in particular that $X(t)$ is close to $\mathbb E[\hat x(t)]$ in any finite time interval.  
Moreover, in the zero-noise limit, the quantity $\mathbb E[\hat x^3(t)]$, which appears in the equation of $\mathbb E[\hat x (t)]$ obtained by taking the mean on both sides of Eq.~\eqref{eq:FitzHugh-Nagumo_Mean_Field},  can be replaced by $(\mathbb E[\hat x(t)])^3$. 
Furthermore, from the numerical simulations described above, we observe that the return trajectories following a perturbation tend to evolve approximately parallel to the horizontal segment of the limit cycle, as illustrated in Fig.~\ref{fig:return_times2}a. Motivated by this observation, we assume that the returning trajectories are approximately horizontal; that is, $y(t) = Y_0$, and consequently $\frac{dy}{dt} = 0$ along these paths.

Summarizing the above considerations, starting from Eq.~\eqref{eq:FitzHugh-Nagumo_Mean_Field} and considering $\sigma = 0$ as the average input current, the equation approximating the mean return trajectory $X(t)$ simplifies to

\begin{equation}
\frac{dX}{dt} = (\xi + \epsilon)X + \delta X^3 + \nu Y_0 + \mu.
\end{equation}

The corresponding return time is then obtained by integration:
\begin{equation} \label{eq:return_time_integral}
    T_r = \int_{X_0}^{X_f} 
    \frac{1}{(\xi + \epsilon)X + \delta X^3 + \nu Y_0 + \mu} \, \mathrm{d}X,
\end{equation}
which represents the quantity to be maximized with respect to $X_0$ and $Y_0$.

We note that, once $Y_0$ is fixed, the maximum return time is achieved when the denominator of the integrand in Eq.~\eqref{eq:return_time_integral} is minimized. This occurs near the nullcline defined by $\dot{X} = 0$.

To study this, we consider the cubic function
\begin{equation}
    f(X) = \delta X^3 + aX + b,
    \label{eq:f_def}
\end{equation}
where $a = \xi + \epsilon$ and $b = \nu Y_0 + \mu$. The corresponding cubic equation
\begin{equation}
    f(X) = 0 \quad \Longleftrightarrow \quad X^3 + pX + q = 0,
    \label{eq:cubic_standard_form}
\end{equation}
is written in standard form by defining
\[
p = \frac{a}{\delta}, 
\qquad 
q = \frac{b}{\delta}.
\]
The discriminant of this cubic equation is
\[
\Delta = -4p^3 - 27q^2.
\]
In the parameter regime considered here, we find $\Delta > 0$ (see Fig.~\ref{fig:return_times2}), indicating that Eq.~\eqref{eq:cubic_standard_form} admits three distinct real roots $r_k$, with $k = 1, 2, 3$.

We can then express the reciprocal of $f(X)$ via partial fraction decomposition as
\begin{align}
    \frac{1}{f(X)} 
    &= \sum_{k=1}^3 \frac{A_k}{X - r_k}, \\
    A_k 
    &= \frac{1}{f'(r_k)} 
     = \frac{1}{3\delta r_k^2 + a}.
    \label{eq:Ak_def}
\end{align}

Integrating term by term, the return time in Eq.~\ref{eq:return_time_integral} becomes
\begin{equation}
    T_r 
    = \int_{X_0}^{X_f} \frac{1}{f(X)} \, \mathrm{d}X
    = \sum_{k=1}^3 
      \frac{1}{f'(r_k)}
      \ln \left| \frac{X_f - r_k}{X_0 - r_k} \right|.
    \label{eq:return_time_roots}
\end{equation}

To estimate the return time from a given perturbation point $(X_0, Y_0)$, one can compute the roots of Eq.~\eqref{eq:f_def} and substitute them into Eq.~\eqref{eq:return_time_roots}, setting $X_f$ to the value of the limit cycle corresponding to $Y = Y_0$. 

The results of the estimation of the return times from Eq.~\eqref{eq:return_time_roots} are shown in Fig.~\ref{fig:return_times2} and can be compared with the numerical results (simulations) shown in Fig. \ref{fig:return_times2}d. It can be observed that, overall, the analytical results are in good agreement with the simulations, showing, in particular, that the region of maximal return times is around the $\frac{dx}{dt}=0$ nullcline. This result can also be derived from Eq.~\eqref{eq:return_time_integral}. Indeed, as it is clear from this equation, in order to maximize the return time, the denominator of the integrand function has to be minimal. This corresponds to perturbing the system towards a point $(X_0, Y_0)$ on the nullcline $\frac{dx}{dt}=0$. It is important to note that the argument just exposed is independent of the values given to the parameters of the system (as long as those parameters are within a region where the approximation used above is valid, that is, trajectories are approximately parallel to the $x$-axis). In other words, this argument is almost independent of the type of equation used, as long as it has a limit-cycle and almost horizontal trajectories.

\begin{figure}[H]
    \centering
    \includegraphics[width=1\linewidth]{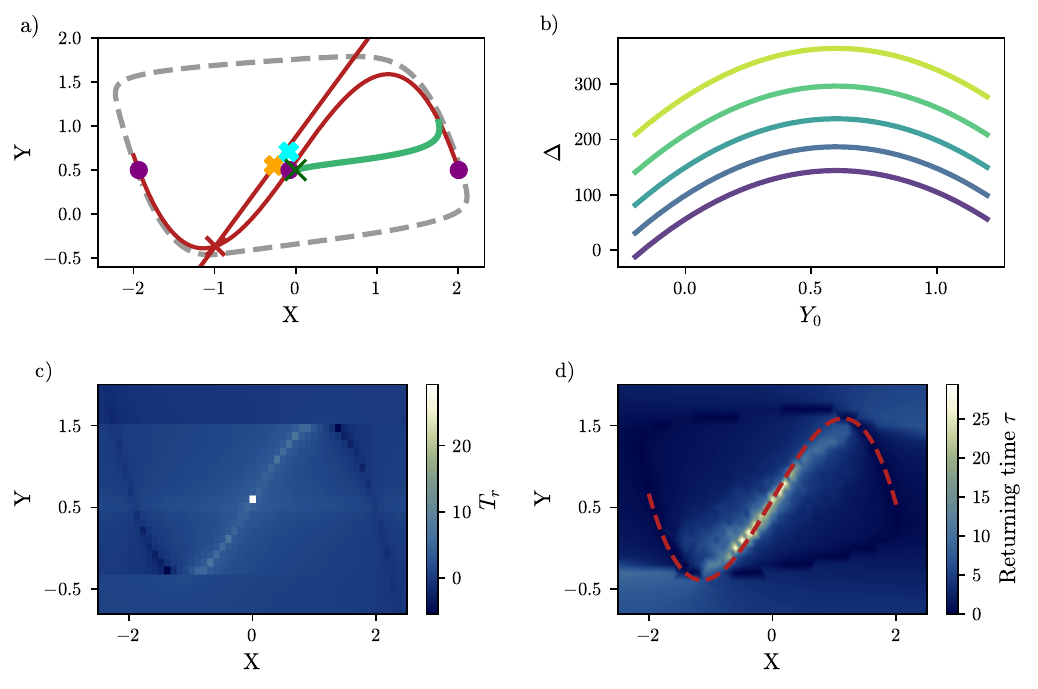}
    \caption{\textbf{Estimation of the return times from Eq.~\eqref{eq:return_time_roots}}.
    a) Limit cycle of $(X,Y)$ and key points: trajectory of the limit cycle (dashed gray) with the nullclines $y = f(x)$ and $y = g(x)$ shown in firebrick. Key points include the baricenter (orange X), geometric center (cyan X), unstable points (red X), an example of a return trajectory from the initial condition $(X_0, Y_0)$ in green, and the control point (dark green X). The roots of the the cubic polynomial Eq.~\eqref{eq:cubic_standard_form} are shwon in purple. \\
    b) Discriminant of the cubic nullcline: The discriminant $\Delta$ of the cubic equation as a function of the initial condition $Y_0$ for different values of $\epsilon \in \{ 0.1, 0.2, 0.3, 0.4, 0.5\}$. \\
    c) Approximate estimation of return times obtained by solving Eq.~\eqref{eq:return_time_roots}, displayed as a heatmap with the nullcline $y = f(x)$ overlaid. \\
    d) Numerical results for the return times obtained from extensive simulations for $\epsilon = 0.3$, displayed as a heatmap. The nullcline $dx/dt = 0$ is shown in red.}
    \label{fig:return_times2}
\end{figure}

\section{Discussion}

We have investigated pulse feedback control of synchronization in a FitzHugh-Nagumo neural population model. We have proved, using the propagation of chaos property, that the neural population equations can be reduced to a single, effective equation in the thermodynamic limit (large number of neurons), from which one can determine special points in phase space
to be subsequently used for pulse control. These points are the centroid, the mean-field limit cycle for $\epsilon = 0$, and the unstable point.

On the basis of numerical simulations, we can argue in favor of a control procedure that leads the system towards the centroid. In the case of bi-variate electrophysiological recordings of the brain (knowledge of $X$ and $Y$, a case that represents a realistic scenario in neuroscience, as argued in the introduction), the centroid can be easily calculated from data, after one cycle of the oscillatory process. In applications, this time may be very short, so the control procedure can be started soon after the establishment of an oscillatory process, such as an epileptic seizure or the onset of tremor in Parkinson's disease. 

In the previous literature, the mean-field cycle at $\epsilon = 0$ was used as a target in the control paradigm \cite{rosenblum2020controlling}. The idea in that work is that pulses are given at specific phases, that is, times on the limit cycle, such as to push the system (with pulses) as close as possible to the mean-field cycle at $\epsilon = 0$. Here we argue that the centroid position for small values of $\epsilon$ is very close to the centroid at $\epsilon = 0$, and it lies within the mean-field cycle at $\epsilon = 0$ (Fig.~\ref{fig:FitzNagumo}). Therefore, in practical terms, using the centroid position is almost equivalent to using the mean-field cycle at $\epsilon = 0$.  

Furthermore, the centroid coordinates practically do not change with changes in the coupling constant ($\epsilon$) (see Fig.~\ref{fig:FitzNagumo}). This is important because in neural tissue it would be expected that the self-coupling of an oscillatory network would change slightly over time. The robustness of the centroid to changes in coupling is, therefore, an advantageous property in terms of control. 

Finally, the most important property for the use of the centroid as a target point for pulse control is that it lies within the region of maximal values of return times (see Fig.~\ref{fig:return_times}). Therefore, if the system is pushed towards the centroid, it takes a relatively long time for it to return to the limit cycle, compared to other points inside the limit cycle. Why not use other points on the $x$-nullcine, which was found to coincide with the region of largest return times? The reason is that the $x$-nullcine is not known, as the equations of motion of the underlying oscillating system are not known (whereas the centroid can be calculated from data).

The above results may appear counterintuitive in that one may think that the unstable equilibrium point of system \eqref{eq:FitzHugh-Nagumo_Mean_Field} (where the velocity field is zero) is the best point to achieve long-lasting desynchronization.  
On the contrary, our analysis supports that the centroid is a better point for all considered values of $\epsilon$. 
Indeed, the unstable point (in red in each of the panels in Fig.~\ref{fig:return_times}) is both out of the region of long return times and is also not known, again, because the equations of motion are not known. 

In the current article, we do not inquire into the technicalities associated with the practical application of the current procedure in neural systems. Those factors were already extensively discussed in the literature (see, for example \cite{rosin2011closed, rosenblum2020controlling}), and we do not need to repeat the complete feedback control procedure, that we take to be identical to the one detailed in \cite{rosenblum2020controlling}, except for the choice of the control point target (in our case, the centroid is easily calculated because our system is bi-variate). 

In summary, our work provides a principled implementable pulse-feedback strategy for desynchronizing oscillatory neural populations in the FitzHugh--Nagumo mean-field framework: by leveraging the thermodynamic-limit reduction and the geometry of return times, we identify the centroid as a robust data-accessible control target that yields consistently long excursions away from the synchronous cycle across coupling strengths. These results clarify why pushing trajectories toward the centroid outperforms targeting the unstable equilibrium or phase-specific points on the $\epsilon=0$ mean-field cycle, and they support centroid-based feedback as a natural candidate for real-time control in settings where only bi-variate recordings are available.\\

\textbf{Code Availability}
The code used to generate the results in this study is publicly available at: \href{https://github.com/MarcoZenari/Pulse-Desynchronization-Neural-Populations}{GitHub repository}.

\textbf{Acknowledgments}
S. A. and S. S. acknowledge the funding provided by the Lincoln project of the INFN 

\bibliography{bibliography}

\clearpage
\appendix

\begin{appendices}

\section{Propagation of chaos}\label{app}

In this Appendix we prove that system \eqref{eq:FitzHugh-Nagumo} has the \textit{propagation of chaos property.}
The proof is an adaptation of the ideas in \cite{touboul12, nodea}.

\begin{theorem}[Propagation of chaos]
    For any $N\geq 1$, let $(I_i)_{i=1}^N$ be independent standard Gaussian random variables and $\{(x_{i,0}, y_{i,0})\}_{i=1}^N$ be i.i.d. with finite second moment and independent of $(I_i)_{i=1}^N$. Consider the particle system 

    \begin{equation*}
    \left\{
    \begin{aligned}
    \frac{d x_i}{dt} & = \xi x_i + \delta x_i^3 + \nu y_i + \epsilon X + \mu + \sigma I_i, \\
    \frac{d y_i}{dt} & = \alpha ( x_i + \beta y_i + \gamma ),\\
    x_i ( 0) & = x_{i,0}, \\
    y_i(0)  &= y_{i,0},
    \end{aligned}
    \qquad i=1,\ldots,N,
    \right.
    \end{equation*}
    and $N$ independent copies $(\hat x_i, \hat y_i)_{i=1}^N$ of the solution to system \eqref{eq:FitzHugh-Nagumo_Mean_Field}, that is,  
    \begin{equation*}
    \left\{
    \begin{aligned}
    \frac{d\hat{x}_i}{dt} & = \xi \hat{x}_i + \delta \hat{x}_i^3 + \nu \hat{y}_i + \epsilon \mathbb{E}[\hat{x}_i] + \mu + \sigma I_i, \\
    \frac{d\hat{y}_i}{dt} & = \alpha ( \hat{x}_i + \beta \hat{y}_i + \gamma ),
    \\
    \hat x_i(0) & = x_{i,0}, \\
    \hat y_i(0) & = y_{i,0},
    \end{aligned}
    \qquad i=1,\ldots,N .
    \right.
    \end{equation*}
    Then, for any $T>0$ and $i=1, \ldots, N$,
    \begin{equation}\label{ss}
        \begin{split}
            \lim_{N\to +\infty} \mathbb{E}\left[\sup_{t\in [0,T]}|x_i(t) - \hat{x}_i(t)| + \sup_{t\in [0,T]}|y_i(t) - \hat{y}_i(t)| \right] = 0.
        \end{split}
    \end{equation}
    \end{theorem}

\begin{proof}

Set 
\[
\Delta(t) = \sup_{s\in [0,t]}|x_i(s) - \hat{x}_i(s)| + \sup_{s\in [0,t]}|y_i(s) - \hat{y}_i(s)| .
\] 
It suffices to prove that there exists a positive constant $C$ (independent of $N$ but possibly depending on $T$) such that  
    \begin{equation*}
        \begin{split}
            \mathbb{E}\left[\Delta(T) \right] \leq C \int_0^T \mathbb{E}\left[ \Delta(t)\right]  dt + \frac{{C}}{\sqrt{N}},
        \end{split}
    \end{equation*}
to conclude, by the Gronwall's lemma, 
\begin{equation}\label{eq:PoC_final_statement}
    \begin{split}
        \mathbb{E}\left[ \Delta(T)\right] \leq \frac{C}{\sqrt{N}}. 
    \end{split}
\end{equation}
The result \eqref{ss} follows by taking the limit $N\to +\infty$ in \eqref{eq:PoC_final_statement}.

\noindent
To start with, we consider the difference 
\[
    y_i(t) - \hat y_i(t) = \alpha \int_0^t (x_i(s) - \hat x_i(s)) ds  
    + \alpha \beta \int_0^t (y_i(s) - \hat y_i(s)) ds,
\]
where we used $y_i(0) = \hat y_i(0)$. 
Taking the absolute value and the supremum over $t \in [0,T]$ gives
\begin{equation}\label{eq:PoC_y_step1}
\sup_{t \in [0,T]} |y_i(t) - \hat y_i(t)|
\le C_1 \int_0^T \Delta(t) dt .
\end{equation}

\noindent 
For the difference $x_i(t) - \hat x_i(t)$, we use the identity $a^3 - b^3 = (a-b) (a^2 + ab + b^2)$ to write 
\begin{equation}\label{eq:diff_x_PoC}
    \begin{split}
        x_i(t) - \hat{x}_i(t) & = \int_0^t (x_i(s) - \hat{x}_i(s)) \left( \xi + \delta (x^2_i(s) + x_i(s) \hat{x}_i(s) + \hat{x}^2_i(s) ) \right) ds \\
        & + \int_0^t \left[\nu (y_i(s) - \hat{y}_i(s)) + \epsilon (X(s) - \mathbb{E}[\hat{x}_i(s)] ) \right] ds ,
    \end{split}
\end{equation}
where again we used $x_i(0) = \hat x_i(0)$. 

\noindent Setting $\phi(t) =x_i(t) - \hat{x}_i(t) $, $g(s) = \xi + \delta (x^2_i(s) + x_i(s) \hat{x}_i(s) + \hat{x}^2_i(s) ) $ and $f(s) = \nu (y_i(s) - \hat{y}_i(s)) + \epsilon (X(s) - \mathbb{E}[\hat{x}_i(s)] )$, 
Eq.~\eqref{eq:diff_x_PoC} 
is of the form $\phi(t) = \int_0^t \phi(s) g(s) ds + \int_0^t f(s) ds$, which has solution  $\phi(t) = \phi(0) + \int_0^t f(s) e^{\int_s^t g(r) dr}ds$.

\noindent Noticing that $\phi(0)= 0$ and 
taking the absolute value, we obtain 
\begin{equation*}
    \begin{split}
        | x_i(t) - \hat{x}_i(t) | & \leq  \int_0^t \left| \nu (y_i(s) - \hat{y}_i(s)) + \epsilon (X(s) - \mathbb{E}[\hat{x}_i(s)] ) \right| e^{\int_s^t \left(\xi + \delta (x^2_i(r) + x_i(r) \hat{x}_i(r) + \hat{x}^2_i(r) )\right) dr } .
    \end{split}
\end{equation*}
Now, since $a^2 + ab + b^2 \geq 0$ and $\delta<0$, we have the inequality $e^{\int_s^t \left(\xi + \delta (x^2_i(r) + x_i(r) \hat{x}_i(r) + \hat{x}^2_i(r) )\right) dr }\leq  e^{\xi T}$ for all $0< s < t<T$. 
Overall, 
\begin{equation}\label{eq:PoC_x_step1}
    \begin{split}
        \sup_{t\in[0,T]}| x_i(t) - \hat{x}_i(t) | & \leq e^{\xi T} \left( \int_0^T |\nu| |y_i(s)-\hat{y}_i(s)| ds + \int_0^T \epsilon | X(s) - \mathbb{E}[\hat{x}_i(s)] | ds \right)\\
        & \leq C_2 \left( \int_0^T \sup_{r\in[0,s]} |y_i(r)-\hat{y}_i(r)| ds + \int_0^T  | X(s) - \mathbb{E}[\hat{x}_i(s)] | ds\right) .
    \end{split}
\end{equation}

\noindent We next observe that for the last summand in \eqref{eq:PoC_x_step1} it holds 
\begin{equation*}
    \begin{split}
        | X(s) - \mathbb{E}[\hat{x}_i(s)] | & \leq \left| \frac{1}{N}\sum_{j=1}^N x_j(s) - \frac{1}{N}\sum_{j=1}^N \hat{x}_j(s)  \right| + \left| \frac{1}{N}\sum_{j=1}^N \hat{x}_j(s) - \mathbb{E}[\hat{x}_i(s)] \right| \\
        & \leq \frac1N \sum_{j=1}^N | x_j(s) - \hat{x}_j(s)| + \left| \frac{1}{N}\sum_{j=1}^N \hat{x}_j(s) - \mathbb{E}[\hat{x}_i(s)] \right| .
    \end{split}
\end{equation*}
Overall, inserting this last inequality  back into \eqref{eq:PoC_x_step1}, we have 
\begin{equation}\label{eq:PoC_x_step2}
    \begin{split}
        \sup_{t\in[0,T]}| x_i(t) - \hat{x}_i(t) | & \leq C_3 \left( \int_0^T \sup_{r\in[0,s]} |y_i(r)-\hat{y}_i(r)| ds + \frac1N \sum_{j=1}^N \int_0^T \sup_{r\in[0,s]} | x_j(r) - \hat{x}_j(r)| ds \right.\\
        & \left. +  \int_0^T \left| \frac{1}{N}\sum_{j=1}^N \hat{x}_j(s) - \mathbb{E}[\hat{x}_i(s)] \right| ds\right) .
    \end{split}
\end{equation} 

\noindent 
To conclude: by symmetry,  $\mathbb{E}\left[\frac{1}{N}\sum_j | x_j(s) - \hat{x}_j(s)|\right] = \mathbb{E}[| x_j(s) - \hat{x}_j(s)|]$ for any $j$, and, since $(\hat x_j(s))_{j=1}^N$ are i.i.d.,  $\mathbb E\left[ \left| \frac{1}{N}\sum_{j=1}^N \hat{x}_j(s) - \mathbb{E}[\hat{x}_i(s)] \right| \right] \leq C_4/\sqrt{N}$ by the central limit theorem.

\noindent Summing \eqref{eq:PoC_y_step1} and \eqref{eq:PoC_x_step2} 
and taking the expectation on both sides, we obtain \eqref{eq:PoC_final_statement}. 

\end{proof}

\end{appendices}

\end{document}